\documentclass[11pt]{article}

\usepackage{graphicx}

\usepackage{amsmath}
\usepackage{amsfonts}
\usepackage{amsthm}
\usepackage{amssymb}
\usepackage{MnSymbol}
\usepackage{bbm}
\usepackage{cite}
\usepackage{enumerate}
\usepackage{authblk}

\usepackage{amsmath}
\usepackage{amsfonts}
\usepackage{amsthm}
\usepackage{amssymb}
\usepackage{MnSymbol}
\usepackage{bbm}
\usepackage{fancyhdr}
\usepackage{algorithmic}
\usepackage{listings}
\usepackage{enumerate}
\usepackage{cite}

\newcommand{\R}{\mathbb{R}}
\newcommand{\N}{\mathbb{N}}

\newcommand{\E}{\mathbb{E}}

\renewcommand{\P}{\mathbb{P}}

\newcommand{\normal}{\mathcal{N}}

\DeclareMathOperator{\argmin}{arg\ min}

\newcommand{\eqD}{\stackrel{d}{=}}

\newtheorem{theorem}{Theorem}[section]

\newtheorem{lemma}[theorem]{Lemma}
\newtheorem{corollary}[theorem]{Corollary}
\newtheorem{proposition}[theorem]{Proposition}

\newtheorem{definition}[theorem]{Definition}

\newtheoremstyle{example}{\topsep}{\topsep}%
     {}
     {}
     {\bfseries}
     {}
     {\newline}
     {\thmname{#1}\thmnumber{ #2}\thmnote{ #3}}

\theoremstyle{example}

\newcommand{\eps}{\epsilon}

\title{On extracting common random bits \\ from correlated sources on large alphabets}
\author[2]{Siu On Chan\thanks{Supported by NSF grant DMS-1106999 and DOD ONR grant N000141110140}}
\author[1,2]{Elchanan Mossel\thanks{Supported by NSF grant DMS-1106999 and DOD ONR grant N000141110140}}
\author[1]{Joe Neeman\thanks{Supported by NSF grant DMS-1106999 and DOD ONR grant N000141110140}}
\affil[1]{Department of Statistics, UC Berkeley}
\affil[2]{Department of Computer Science, UC Berkeley}

\begin{document}
\maketitle

\begin{abstract}

Suppose Alice and Bob receive strings $X=(X_1,\ldots,X_n)$ and $Y=(Y_1,\ldots,Y_n)$ 
each uniformly random in $[s]^n$ but so that $X$ and $Y$ are correlated . For each symbol $i$, we have that 
$Y_i = X_i$ with probability $1-\eps$ and otherwise $Y_i$ is chosen independently and uniformly from $[s]$. 

Alice and Bob wish to use their respective strings to extract a uniformly chosen common sequence from $[s]^k$  but without communicating. How well can they do? The trivial strategy of outputting the first $k$ symbols yields an agreement probability of 
$(1 - \eps + \eps/s)^k$. In a recent work by Bogdanov and Mossel it was shown that in the binary case where $s=2$ and $k = k(\eps)$ is large enough then it is possible to extract $k$ bits with a better agreement probability rate. 
In particular, it is possible to achieve agreement probability 
$(k\eps)^{-1/2} \cdot 2^{-k\eps/(2(1 - \eps/2))}$ using a random construction based on Hamming balls,
and this is optimal up to lower order terms.

In the current paper we consider the same problem over larger alphabet sizes $s$ and we show that
the agreement probability rate changes dramatically as the alphabet grows. In particular we show 
no strategy can achieve agreement probability better than $(1-\eps)^k (1+\delta(s))^k$ where $\delta(s) \to 0$ as $s \to \infty$. 
We also show that Hamming ball based constructions have {\em much lower} agreement probability rate than the trivial algorithm 
as $s \to \infty$. Our proofs and results are intimately related to subtle properties of hypercontractive inequalities. 
\end{abstract}

\section{Introduction}
For an integer $s \ge 2$, consider two $[s]^n$-valued random variables
$X, Y$ (where $[s] = \{0, 1, \dots, s-1\}$)
which are sampled by first choosing $X$ uniformly and then,
independently for every coordinate $i$, taking $Y_i$ to be
a copy of $X_i$ with probability $1-\epsilon$ and an independent
sample from $[s]$ otherwise.
We will write $\P_\epsilon$ for this joint distribution on $X$
and $Y$. Note that $X$ and $Y$ are both uniformly distributed in $[s]^n$. 

The non-interactive correlation distillation (NICD) is defined as follows:
suppose that one party (Alice) receives $X$ and another (Bob) receives $Y$. Without
any communication, each party chooses a string that is uniformly distributed in $[s]^k$ with the goal of maximizing the probability 
that the two strings chosen by Alice and Bob are identical.  

\subsection{Motivation and Related Work}
This problem was studied in~\cite{BogdanovMossel:11} in the case $s=2$, with motivation from various areas.
One major motivation comes from the goal of extracting a unique identification string from process
variations~\cite{LLGSVD:05,YLHMGZ:09}, particularly in a noisy setup~\cite{SHO:08}. 

The case where the goal of the two parties is to extract a single bit was studied independently
a number of times; in this case the optimal protocol is for the two parties
to use the first bit. See~\cite{Yang:07} for references and for studying the problem of extracting one bit from two correlated sequences with different correlation structures.

In~\cite{MoOdOl:05,MORSS:06} a related question is studied: if $m$
parties receive noisy versions of a common random string, where the noise
of each party is independent, what is the strategy for the $m$ parties that maximizes
the probability that the parties agree on a {\em single} random bit of output
without communicating? \cite{MoOdOl:05} shows that for large $m$ using the majority functions on all
bits is superior to using a single bit and \cite{MORSS:06} uses hypercontractive inequalities to show that for large $m$, majority is close to being optimal. Both results were recently extended to general string spaces in~\cite{MoOlSe:12}. 

For any $k \in \N$, one protocol -- which we will call the
``trivial protocol'' -- is
for both parties to take the first $k$ symbols of their strings.
The success probability of this protocol is
$(1 - (1 - \frac 1s)\epsilon)^k \approx \exp(-k\epsilon (1-\frac 1s))$.
When $s=2$ and the protocol outputs a single bit (ie.\ $k=1$),
it is known (see e.g.~\cite{MoOdOl:05}) that the optimal
protocol is for both parties to choose the first bit. For larger $k$,
this is no longer true. Bogdanov and Mossel~\cite{BogdanovMossel:11}
studied the case $s=2$, and showed that any protocol which outputs
a uniformly random length-$k$ string has a success probability
of at most $\exp(-k \epsilon (\ln 2)/2)$.
In other words, if $p$ is the success probability of the trivial algorithm
for choosing a $k$-bit string, then every protocol with success probability
at least $p$ emits at most $k / \ln 2$ bits.

Bogdanov and Mossel showed that their bound was sharp by providing an example
(for a restricted range of $\epsilon$ and $k$) with
success probability which, for any $\delta > 0$, is at least $\exp(-k\epsilon(1+\delta) / 2)$
for small $\epsilon$ and large $k$.
In other words, if $p$ is the success probability of the trivial algorithm for choosing
a $k$-bit string, then they gave a protocol that succeeds with probability $p$
and produces a string of length $k / ((1+\delta)\ln 2)$.
Their construction was built by taking random translations of Hamming balls; we will return
to it in more detail later.

\subsection{Our results} 
We study an extension of the upper bound of~\cite{BogdanovMossel:11} to a larger alphabet.
In our main result we show that in the case of large alphabets, 
the constant-factor gap
between the upper bound and the performance of the trivial algorithm vanishes;
hence, the trivial algorithm is almost optimal for large alphabets. In particular we show 
no strategy can achieve agreement probability better than $(1-\eps)^k (1+\delta(s))^k$ where $\delta(s) \to 0$ as $s \to \infty$.

We then turn to analyze generalizations of the Hamming ball based construction of~\cite{BogdanovMossel:11}. 
Interestingly we show that these have {\em much lower} agreement probability rate than the trivial algorithm 
as $s \to \infty$. 

In this respect it is interesting to compare the case of a large number of parties
that extract a single symbol to the case 
of two parties who extract a longer string.
In the first case, the results of~\cite{MoOlSe:12} generalize those of 
~\cite{MoOdOl:05,MORSS:06} to show that Hamming ball based protocols are almost optimal for all values of $s$ when the number of parties $m$ is large. In the case presented here, Hamming ball type constructions quickly deteriorate as $s$ increases and the trivial protocol becomes almost optimal. 

The difference between the two phenomena may be explained by the fact that the problem studied
in~\cite{MoOdOl:05,MORSS:06} is closely related to reverse-hypercontractive inequalities
which hold uniformly in $s$~\cite{MoOlSe:12},
while the problem studied here is closely related to hypercontractive
inequalities which deteriorate as $s$ increases.

Our results show that the trivial algorithm is optimal up to a factor of 
$(1+\delta(s))^k$ where $\delta(s) \to 0$ as $s \to \infty$. 
An interesting open problem is to find an almost optimal algorithm for large $s$, i.e., an 
algorithm whose agreement probability is provably optimal up to a factor of $2^{-o(k)}$. 
It is quite possible that  the trivial protocol is optimal
for some large fixed values of $s$  and all large enough $k$.


\section{Definitions and results}

A {\em protocol} for NICD is defined by two
functions $f, g: [s]^n \to [s]^*$. Upon receiving their strings $X, Y \in [s]^n$,
the two parties compute $f(X)$ and $g(Y)$ respectively.
The protocol is successful if both parties agree on the same output; that is,
if $f(X) = g(Y)$. Therefore, finding an optimal NICD algorithm is
equivalent to finding functions $f, g: [s]^n \to [s]^*$ which maximize
$\P_\epsilon(f(X) = g(Y))$.

In the introduction, we mentioned the requirement that $f$ and $g$ are uniformly
distributed on $[s]^k$. In fact, we will require less for our negative results and guarantee more
in our positive results. 
In particular, for our negative results,
we will only assume that $f$ and $g$ have min-entropy at most $k$, meaning that
$\P(f(X) = z) \le s^{-k}$ for all $z \in [s]^*$ and similarly for $g$.
Of course, if $f: [s]^n \to [s]^k$ is uniformly distributed then it
has min-entropy $k$.

\subsection{Reduction to a question about sets}

Using an observation of~\cite{BogdanovMossel:11}, we can reduce the NICD problem
to the problem of finding a sets $A \subset [s]^n$ which
maximize $\P_\epsilon(Y \in A | X \in A)$. On the one hand, if we are given
good functions $f$ and $g$ then we can find a set
$A$ such that $\P(Y \in A | X \in A)$ is large:

\begin{theorem}\label{thm:partition-to-set}
  For any functions $f, g: [s]^n \to [s]^*$ having min-entropy $k$
  there is a set $A \subset [s]^n$ with
  $|A| \le s^{n-k}$ such that
  for every $0 \le \epsilon \le 1$,
  \[
    \P_\epsilon(Y \in A | X \in A) \ge \P_\epsilon(f(X) = g(Y)).
  \]
\end{theorem}

On the other hand, if we have a good set $A$ then we can construct
a function $f$ by taking certain translates of $A$.

\begin{theorem}\label{thm:set-to-partition}
  If $A \subset [s]^n$ with
  $\frac{1}{8} s^{n-k} \le |A| \le \frac{1}{4} s^{n-k}$ then
  there is a function $f: [s]^n \to [s]^k$ such that
  \begin{enumerate}
    \item $f(X)$ is uniformly distributed on $[s]^k$
    \item $f(X)$ is uniformly distributed on $[s]^k$ conditioned on $f(X) = f(Y)$
    \item
      for every $0 \le \epsilon \le 1$,
      \[
        \P_\epsilon(f(X) = f(Y)) \ge \frac{1}{16} \P_\epsilon(Y \in A | X \in A).
      \]
  \end{enumerate}
\end{theorem}

Note that the $f$ that we produce in Theorem~\ref{thm:set-to-partition}
satisfies stronger requirement than the one that we require in Theorem~\ref{thm:partition-to-set}.
Indeed, the $f$ from Theorem~\ref{thm:set-to-partition} is uniformly distributed
instead of only having a small minimum entropy. Moreover, $f(X)$ is uniformly
distributed given $f(X) = f(Y)$, which means that a successful execution of the
protocol will result in the two parties having uniformly random strings.

\subsection{Negative results on the performance of NICD}

In view of Theorems~\ref{thm:partition-to-set} and~\ref{thm:set-to-partition}, the NICD problem reduces
to the study of $\P_\epsilon(Y \in A | X \in A)$ over sets $A \subset [s]^n$
with a given cardinality. Actually, it turns out to be more convenient to
normalize the cardinality instead of restricting it:

\begin{definition}
For $A \subset [s]^n$, define
\[
M_\epsilon(A) = \frac{\ln \P_\epsilon(Y \in A | X \in A)}{\ln \P(A)}.
\]
\end{definition}

To illustrate the definition,
consider the set $A = \{x: x_1 = \dots = x_k = 0\}$, which corresponds
to the trivial algorithm that selects the first $k$ symbols.
In this case, $\P_\epsilon(Y \in A | X \in A) =
(1-(1-s^{-1}) \epsilon))^k$. Since $\P(A) = s^{-k}$, it follows
that
\begin{equation}\label{eq:trivial-example}
M_\epsilon(A) = \frac {1}{\ln s}
\ln \left(\frac{1}{1-(1-s^{-1})\epsilon}\right).
\end{equation}

Our main result is that the above example is optimal as
$s \to \infty$.

\begin{theorem}\label{thm:upper-bound}
For every $\delta, \epsilon > 0$ there exists $S < \infty$
such that for all $n \in \N$ and all $s \ge S$,
any set $A \subset [s]^n$ satisfies
\[
  M_\epsilon(A) \ge \frac{1}{\ln s} \Big(\ln\frac{1}{1-\epsilon} - \delta\Big)
\]
\end{theorem}

Note that since $\ln \P(A)$ is negative, Theorem~\ref{thm:upper-bound}
provides an upper bound on $\P_\epsilon(Y \in A \mid X \in A)$
for all sets $A$ of a fixed probability, and therefore an upper bound
on the agreement probability of any NICD protocol.
We remark that our proof extends to the case where the
$X_i$ are chosen independently from some distributions whose smallest
atoms are at most $\alpha$. In this case, the theorem holds with
$s$ replaced by $1/\alpha$.

As a corollary of Theorems~\ref{thm:partition-to-set}
and~\ref{thm:upper-bound}, we obtain a bound on the performance
of any NICD protocol.

\begin{corollary}\label{cor:nicd-fixed-k}
  For any $\delta, \epsilon > 0$, there exists $S < \infty$ such that
  for all $n, k \in \N$,
  for any $s \ge S$, and for any
  NICD protocol $f, g$ on $[s]$ with min entropy at most $k$,
  the probability that the protocol succeeds with noise $\epsilon$ is at most
  $(1-\epsilon)^ke^{\delta k}$.
\end{corollary}

Since the success rate of the trivial protocol with min-entropy $k$ is
bigger than $(1-\epsilon)^k$, this shows that for large $s$,
no protocol can be succeed with much higher probability than the trivial protocol.

\begin{proof}
  Fix a protocol $f, g$ and let $A$ be a set such that
  $|A| \le s^{n-k}$ and
  $\P_\epsilon(Y \in A | X \in A) \ge \P_\epsilon(f(X) = g(Y))$
  (such an $A$ exists by Theorem~\ref{thm:partition-to-set}).
  Then Theorem~\ref{thm:upper-bound} implies (recalling that
  $\ln \P(A)$ is negative)
  \[\ln\P_\epsilon(Y \in A | X \in A)
    \le \frac{\ln \P(A)}{\ln s} \Big(\log \frac{1}{1-\epsilon} - \delta\Big)
    \le -k\Big(\log \frac{1}{1-\epsilon} - \delta\Big)
  \]
  Taking the exponential of both sides yields the corollary.
\end{proof}

Of course, we can also restate Corollary~\ref{cor:nicd-fixed-k}
for a fixed probability of success and a varying $k$:

\begin{corollary}\label{cor:nicd-fixed-p}
  For any $\delta, \epsilon > 0$, there exists $S < \infty$ such that
  for all $n \in \N$, for all $0 < p < 1$,
  for any $s \ge S$, and for any
  NICD protocol $f, g$ that succeeds with probability at least $p$,
  if $k$ is the min-entropy of the protocol then the trivial protocol on
  $\lfloor k \frac{\log (1-\epsilon)}{\log (1-\epsilon) + \delta}\rfloor$ symbols also succeeds with probability
  at least $p$.
\end{corollary}

In other words, for a fixed probability of failure, a trivial protocol
can recover almost as many symbols as any other protocol (when $s$ is large).

The dependence of $S$ on $\delta$ and $\epsilon$ is not made explicit in
our proof. However, our proof does provide a way to
approximate $S(\delta, \epsilon)$ on a computer; therefore, we produced a plot
(Figure~\ref{fig:sim})
showing the approximate value of $S$
for various values of $\delta$ and $\epsilon$.

\begin{figure}
\begin{center}
\includegraphics[width=0.8\textwidth]{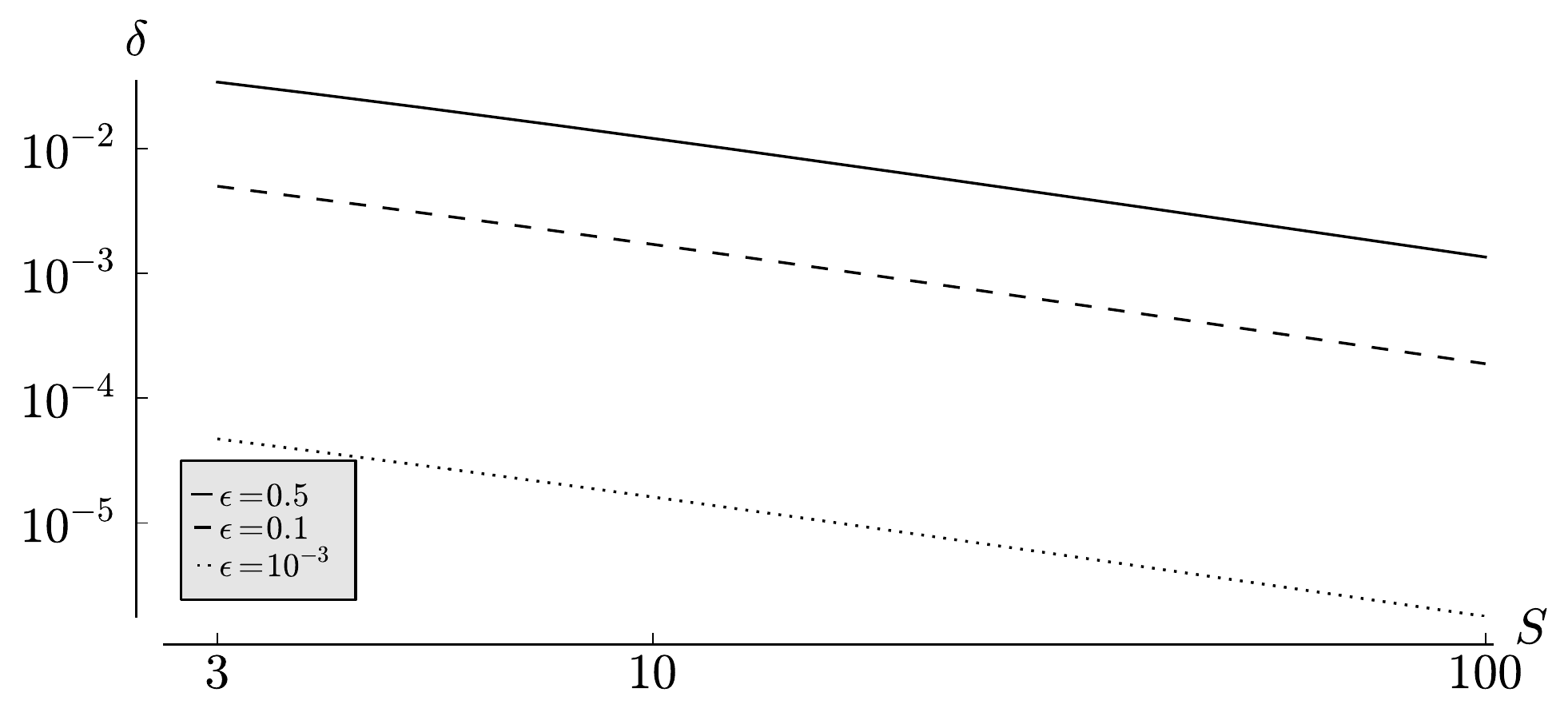}
\caption{
  The relationship, in log-log scale, between $S$ and $\delta$
  in Theorem~\ref{thm:upper-bound} for various values of $\epsilon$: $0.5$ (solid),
  $0.1$ (dashed), and $10^{-3}$ (dotted). For each of these values of $\epsilon$,
  every point $(s, \delta)$ that is above the corresponding line,
  and every $n \in \N$, all sets $A \subset [s]^n$ satisfy
  $M_\epsilon(A) \ge \frac{1}{\ln s} (\ln \frac{1}{1-\epsilon} - \delta)$.
\label{fig:sim}
}
\end{center}
\end{figure}

\subsection{An example: the Hamming ball}

As we have already mentioned,~\cite{BogdanovMossel:11} showed that
when $s=2$, the trivial algorithm is optimal up to a constant factor;
As we have just seen, this constant factor converges to 1 as $s \to \infty$.
However,~\cite{BogdanovMossel:11} also gave a positive result: they gave
an example that achieves optimal performance (at least, up to lower order
terms and for a particular range of $k$ and $\epsilon$).
Since their example can be generalized to $s > 2$, we can
examine its performance as $s \to \infty$, and compare it to the trivial
algorithm.

Define the set
\[
  A_{s,\alpha,n} = \Big\{x \in [s]^n: \#\{i: x_i \ne 0\} \le n \frac{s-1}{s} - \alpha \sqrt n \Big\}.
\]
In other words, $A_{s,\alpha,n}$ is a Hamming ball around zero of radius $n \frac{s-1}{s} - \alpha \sqrt n$.
When $s=2$,~\cite{BogdanovMossel:11} showed that $M_\epsilon(A_{2,\alpha,n}) \approx \epsilon/2$
as $n, t \to \infty$ and $\epsilon \to 0$ (note that this does not contradict
Theorem~\ref{thm:upper-bound}, which only holds for sufficiently large $s$).
Since the trivial algorithm has $M_\epsilon(A) \approx \epsilon / (2 \ln 2)$ for small $\epsilon$,
this shows that the Hamming ball NICD protocol is better than the trivial one for $s=2$.
The situation reverses, however, as $s$ grows:
\begin{proposition}
  There exists a constant $c$ such that for any $s, \alpha$ and $\epsilon$,
  \[
    \lim_{n \to \infty} M_\epsilon(A_{s,\alpha,n}) \ge c \epsilon.
  \]
\end{proposition}
Since the trivial algorithm has $M_\epsilon(A) \sim \epsilon / \ln s$,
it is better than the Hamming ball protocol when $s$ is large.
In terms of the agreement probability, an argument like the proof of Corollary~\ref{cor:nicd-fixed-k}
shows that the agreement probability of the Hamming ball protocol is
at most $(1-\epsilon)^{c k \ln s}$. In terms of the number of recovered symbols,
the Hamming ball protocol with the same agreement probability as the $k$-symbol
trivial protocol can only recover $c k/\ln s$ symbols.

\section{Reduction to a single set}

In this section, we will prove Theorems~\ref{thm:partition-to-set}
and~\ref{thm:set-to-partition}, which
reduce the NICD problem to a question about optimal subsets of $[s]^n$.
The proof of Theorem~\ref{thm:partition-to-set} is straightforward,
and essentially follows directly from the Cauchy-Schwarz inequality.

\begin{proof}[Proof of Theorem~\ref{thm:partition-to-set}]
  Suppose that
  $f, g: [s]^n \to [s]^k$ have min-entropy $k$.
  For $z\in [s]^k$, let $f_z:[s]^n\to \{0,1\}$ be the function
  \[ f_z(x) = \begin{cases} 1, \qquad \text{if $f(x) = z$,} \\ 0, \qquad
      \text{otherwise.} \end{cases} \]
  Define $g_z(x)$ similarly.
  Then
  \begin{align*}
    \P_\epsilon(f(X) = g(Y)) &= \sum_{z\in [s]^k} \P_\epsilon(f(X) = g(Y) = z) \\
    &= \sum_{z\in [s]^k} \E f_z(X) g_z(Y) \\
    &\leq \sum_{z\in [s]^k} \sqrt{\E f_z(X) f_z(Y)} \sqrt{\E g_z(X) g_z(Y)} \\
  &\leq \sqrt{\sum_{z\in [s]^k} \E f_z(X) f_z(Y)} \sqrt{\sum_{z\in [s]^k}
    \E g_z(X) g_z(Y)} ,
  \end{align*}
  where both inequalities are Cauchy-Schwarz.

  For each $z \in [s]^k$, let $A_z$ be the set $f^{-1}(z)$. Since $f$ has min-entropy $k$, $|A_z| \le s^{n-k}$
  for all $z$. Let $A$ be the $A_z$ which maximizes $\P_\epsilon[Y \in A_z \mid X \in A_z]$.
  Then
  \begin{align*}
    \sum_{z\in [s]^k} \E f_z(X) f_z(Y) &= \sum_{z\in [s]^k} \P_\epsilon(f(X) = f(Y) =
    z) \\
    &= \sum_{z\in [s]^k} \P_\epsilon(f(X) = z) \P_\epsilon(Y\in A_z\mid X\in A_z)
    \\
    &\le \P_\epsilon(Y \in A \mid X \in A).
    \qedhere
  \end{align*}
\end{proof}

The idea behind Theorem~\ref{thm:set-to-partition} is,
given a set $A \subset [s]^n$ with $\frac{1}{8} s^{n-k} \le |A| \le \frac{1}{4} s^{n-k}$,
to construct a partition of $[s]^n$ out of randomly translated copies
of $A$. Let $C \subset [s]^n$, $|C| = s^k$ be the set of ``centers.'' We will choose
$C$ randomly; we will say how to choose it later. Let $f_C: [s]^n \to C$
to be some function with the property that if $x \in A + c$ for a unique $c \in C$
then $f_C(x) = c$. Clearly, then,
\begin{equation}\label{eq:agreement-lower-bound}
  \P_\epsilon(f_C(X) = f_C(Y)) \ge 
  \P_\epsilon(\exists ! c \in C \text{ such that } X, Y \in A + c).
\end{equation}
The goal is to find a $C$ which makes the right-hand side large; this will
allow us to prove property 3 in the second part of Theorem~\ref{thm:partition-to-set}.

Note, by the way, that it is sufficient to prove 
Theorem~\ref{thm:set-to-partition} with $[s]^k$ replaced by an arbitrary set $C$
satisfying $|C| = s^k$. Since such a $C$ is in bijection with $[s]^k$, the theorem as stated
will follow.

\begin{lemma}\label{lem:s-to-p-stability}
  Suppose that $C$ is chosen (randomly) such that for any $a, b \in [s]^n$,
  $\P(a, b \in C) = s^{2(t-n)}$. Then
  \[
    \E_C \P_\epsilon (f_C(X) = f_C(Y)) \ge \frac{1}{16} \P(\epsilon(Y \in A \mid X \in A).
  \]
  In particular, there exists a fixed $C$ such that $f_C$ satisfies property 3
  of Theorem~\ref{thm:set-to-partition}.
\end{lemma}

\begin{proof}
  We begin from the right-hand side of~\eqref{eq:agreement-lower-bound}:
  \begin{align}
  &\P_\epsilon(\exists ! c \in C \text{ such that } X, Y \in A + c) \\
  &\geq \E_C \sum_{c\in C} \P_\epsilon (X,Y\in A_c) \left( 1 - \sum_{c'\neq
      c} \P_\epsilon (\text{$X$ or $Y\in A_{c'} \mid X,Y\in A_c$}) \right) \notag \\
      \label{eq:partition-to-set-1}
      &= s^k \E_c \P_\epsilon (X,Y\in A_c) \left( 1 - (s^k-1) \E_{c'} \P_\epsilon
    (\text{$X$ or $Y\in A_{c'} \mid X,Y\in A_c$}) \right) .
  \end{align}
  By our assumption on the distribution of $C$,
  $c'\neq c$ is uniformly random given $c$.
  Thus
  \begin{align*}
    \E_{c'} \P_\epsilon (\text{$X$ or $Y\in A_{c'}\mid X,Y\in A_c$}) &\leq
    2\E_{c'} \P_\epsilon (X\in A_{c'}\mid X,Y\in A_c) \\
    &\leq 2\P_\epsilon (X\in A) \leq s^{-k}/2,
  \end{align*}
  where the last line follows because $|A| \le s^{n-k}/4$.

  Plugging this into~\eqref{eq:partition-to-set-1},
  \[ \E_C \P_\epsilon(f(X) = f(Y)) \geq \frac{s^k}2 \P_\epsilon (X,Y\in A) =
  \frac{\P_\epsilon(Y \in A \mid X \in A)}{16} .
  \qedhere
  \]
\end{proof}

To check properties 2 and 3, we need to be a little more specific about our
choice of $f_C$. So far, we have only assumed that $f_C(x) = c$ if $c$ is the
only member of $C$ with $x \in A + c$. Now, take
$\prec$ to be some total order on $[s]^n$ with the property that $x \prec y$
whenever $x \in A, y \not \in A$. Then define $f_C(x) = \argmin_{c \in C} (x - c)$
(where the arg min is taken with respect to the ordering $\prec$).
This defines $f_C$ on all of $[s]^n$, and it has the property that we required
before:
if $f_C(x) \in A + c$ for a unique $c$, then $f_C(x) - c \in A$ and $f_C(x) - c' \not \in A$
for every $c' \ne c$. By our requirement on $\prec$, $f_C(x) - c \prec f_C(x) - c'$ for every
$c' \ne c$ and so $f_C(x) = c$.

\begin{lemma}\label{lem:s-to-p-invariance}
  If there is a subgroup $G \subset ([s]^n, +)$ and some $a \in [s]^n$ such that
  $C = G + a$, then $f_C$ satisfies properties 1 and 2 of Theorem~\ref{thm:set-to-partition}.
\end{lemma}

\begin{proof}
  For any $g \in G$,
  \[
    f_C(x + g) = \argmin_{c \in C} (x - (c - g)) = g + \argmin_{c \in C - g} (x - c)
    = f_C(x) + g,
  \]
  since $C - g = C$.
  Moreover, note that the distribution of $(X, Y)$ is invariant under translation, in the sense that
  for any fixed $g \in [s]^n$, $(X, Y) + g \eqD (X, Y)$.
  Hence,
  \[
    \P(f(X) = c) = \P(f(X + g) = c) = \P(f(X) = c + g)
  \]
  for any $c \in C, g \in G$. Since $G$ acts transitively on $C$, this implies that
  $\P(f(X) = c) = 1/|C| = s^{-k}$; in other words, $f(X)$ is distributed uniformly on $C$.

  Similarly,
  \[
    \P(f(X) = f(Y) = c) = \P(f(X) = f(Y) = c + g)
  \]
  for any $c \in C, g \in G$ and so $\P(f(X) = f(Y) = c) = s^{-k} \P(f(X) = f(Y))$; in other
  words, $f(X)$ is uniformly distributed on $C$ conditioned on $f(X) = f(Y)$.
\end{proof}

\begin{proof}[Proof of Theorem~\ref{thm:set-to-partition}]
To prove Theorem~\ref{thm:set-to-partition}, we need to find a set $C$ which satisfies
the hypotheses of Lemmas~\ref{lem:s-to-p-stability} and~\ref{lem:s-to-p-invariance}.
In~\cite{BogdanovMossel:11}, they chose $C$ to be a uniformly random $k$-dimensional
affine subspace of $[2]^n$,
but since $[s]^n$ is not a vector space for every $s$, we will need something
slightly more complicated.

Let $s = \prod_{i=1}^m p_i^{j_i}$ be the prime factorization of
$s$. By the Chinese remainder theorem, the group $([s]^n, +)$ is isomorphic
to $\bigoplus_{i=1}^m ([p_i]^{nj_i}, +)$; let $\phi: \bigoplus_i ([p_i]^{nj_i}, +) \to [s]^n$
be an isomorphism. Independently for each $i = 1, \dots, m$
and $j = 1, \dots, k_i$,
let $G_{i,j}$ be a uniformly random $k$-dimensional subspace of $[p_i]^{n}$
(which is a vector space), and let $a_{i,j}$ be a uniformly random element of $[p_i]^n$.
Finally, define
\[
  C = \phi\Big(\bigoplus_{i,j} (a_{i,j} + G_{i,j})\Big)
  = \phi\Big(\bigoplus_{i,j} a_{i,j}\Big) + \phi\Big(\bigoplus_{i,j} G_{i,j}\Big).
\]
Since $\phi(\bigoplus_{i,j} G_{i,j})$ is a subgroup of $[s]^n$, the condition of
Lemma~\ref{lem:s-to-p-invariance} is satisfied with probability 1.

To check the condition of Lemma~\ref{lem:s-to-p-stability}, note that
for any $b = \bigoplus_{i,j} b_{i,j}$ and
$c = \bigoplus_{i,j} c_{i,j}$ in $\bigoplus_{i=1}^m [p_i]^{n k_i}$,
\[\P(b_{i,j}, c_{i,j} \in a_{i,j} + G_{i,j}) = p_i^{2(n-k)}\]
because $G_{i,j}$ is a uniformly random $k$-dimensional subspace of $[p_i]^n$.
Since the $a_{i,j}$ and
$G_{i,j}$ are independent, it follows
that
\[\P(\phi(b), \phi(c) \in C) = \prod_{i,j} P(a_{i,j}, b_{i,j} \in C_{i,j}) = s^{2(n-k)}.\]
That is, the distribution of $C$
satisfies the condition of Lemma~\ref{lem:s-to-p-stability}.
In particular, there exists a non-random $C'$ that belongs to the support of $C$,
and which also satisfies condition 3 of Theorem~\ref{thm:set-to-partition}.
By the previous paragraph, the fact that it belongs to the support of $C$ implies
that it also satisfies conditions 1 and 2.
\end{proof}

\section{An upper bound on agreement}
The proof of Theorem~\ref{thm:upper-bound} uses a hypercontractive
inequality in much the same way as it was used in~\cite{BogdanovMossel:11}.
The difference here is that~\cite{BogdanovMossel:11} used only the hypercontractive
inequality over the two-point space with the uniform measure, while we need one that applies to spaces
with more than two points.
Before stating this hypercontractive inequality, we need to define
the appropriate Bonami-Beckner-type operator: for a function $g: [s] \to \R$,
and some $0 < \tau < 1$,
define $S_\tau g = \tau g + (1 - \tau) \E g$. Thus,
for any $0 < \tau < 1$, and any $1 \le p, q \le \infty$, $S$ is an operator
$L_p([s]) \to L_q([s])$. We define $T_\tau: L_p([s]^n) \to L_q([s])^n$ by
$T_\tau = S_\tau^{\otimes n}$.
The operator $T_\tau$ can also be written in terms of the Fourier
expansion of $f$; see~\cite{Wolff:07} for details.
For us, the crucial property of $T_\tau$ is that
\begin{equation}\label{eq:bonami-beckner}
  \E_\epsilon f(X) f(Y) = \E (T_\tau f)^2
\end{equation}
when $\tau = \sqrt{1-\epsilon}$. This fact was used in~\cite{BogdanovMossel:11}
for $s=2$ to establish Theorem~\ref{thm:upper-bound} in that case.

The following hypercontractive inequality is due to Oleszkiewicz~\cite{Oleszkiewicz:03}:
\begin{theorem}\label{thm:hyper}
Fix $s \in \N$ and set $\alpha = \frac{1}{s}$, $\beta = 1-\alpha$.
Define
\[
\sigma(\alpha, p) = \left(
\frac{\beta^{2-2/p} - \alpha^{2-2/p}}
{\alpha^{1-2/p} \beta - \beta^{1-2/p} \alpha}
\right)^{1/2}.
\]
Then for any $f: [s]^n \to \R$,
if $\tau \le \sigma(\alpha, p)$ then
\[
\|T_\tau f\|_2 \le \|f\|_p.
\]
\end{theorem}

We remark that the reason for not having an explicit $S(\delta)$
in Theorem~\ref{thm:upper-bound} and its corollaries
is that we do not know how to solve for $p$ in terms of $\sigma(\alpha, p)$.
However, an approximate solution can easily be found on a computer, and we used
such an approximation to produce Figure~\ref{fig:sim}.
To obtain Theorem~\ref{thm:upper-bound}, it suffices to study the limit of
$\sigma(\alpha, p)$ as $\alpha \to 0$.
Essentially, $\sigma^2(\alpha, p) \approx \alpha^{1-2/p}$ for small $\alpha$,
and so if we take $p$ to be slightly larger than what is needed to solve
$\alpha^{1-2/p} = 1-\epsilon$, then we will have $\sigma(\alpha, p) \ge 1-\epsilon$.
This will allow us to apply Theorem~\ref{thm:hyper} with $\tau = \sqrt{1-\epsilon}$.

\begin{lemma}\label{lem:sigma}
  Let $p = p(\alpha, \delta, \epsilon)$ solve 
  \[
  \alpha^{(2/p - 1) - \delta/\ln \alpha} = 1-\epsilon.
  \]
  Then for any $\delta > 0$ and $\epsilon^* \in (0, 1)$, there is an $A(\delta, \epsilon^*) > 0$
  such that $\alpha < A(\delta, \epsilon^*)$
  implies that for all $\epsilon \in (0, \epsilon^*)$,
\[\sigma^2(\alpha, p(\alpha, \delta, \epsilon)) \ge 1-\epsilon.\]
\end{lemma}

\begin{proof}
  Note that the definition of $p$ ensures that $p < 2$ for all $\alpha, \delta, \epsilon$.
  By the definition of $\sigma$,
  \begin{equation}\label{eq:sigma-lower-bound}
    \sigma^2(\alpha, p) \alpha^{1-2/p} = \frac{\beta^{2-2/p} - \alpha^{2-2/p}}{\beta - \alpha^{2/p} \beta^{1-2/p}}
    \ge \beta^{2-2/p} - \alpha^{2-2/p}.
  \end{equation}
  Fix $\epsilon^*$ and $\delta$, and note that as $\alpha \to 0$,
  $2 - 2/p \to 1$ uniformly for all $\epsilon \in (0, \epsilon^*)$.
  Hence, the right-hand side of~\eqref{eq:sigma-lower-bound}
  converges to 1 (uniformly in $\epsilon$) as $\alpha \to 0$.
  Plugging in the definition of $p$,
  \[
    \frac{\sigma^2(\alpha, p)}{1-\epsilon}
  = \sigma^2(\alpha, p) \alpha^{1-2/p} \alpha^{-\delta/\ln \alpha}
  \ge (1-o(1)) e^{-\delta}.
  \]
  In particular, the limit of the right hand side is strictly smaller than one, and so
  $\sigma^2(\alpha, p) \ge 1-\epsilon$ for sufficiently small $\alpha$.
\end{proof}

\begin{proof}[Proof of Theorem~\ref{thm:upper-bound}]
Fix $\epsilon, \delta > 0$. Let $A$ and $p$ be as in Lemma~\ref{lem:sigma}
and define $S = 1/A$. If $s \ge S$ then $\alpha = 1/s \le A$ and so
Lemma~\ref{lem:sigma} implies that $\sigma^2(\alpha, p) \ge 1-\epsilon$.
Thus,~\eqref{eq:bonami-beckner} and Theorem~\ref{thm:hyper} imply that
\[
\P_\epsilon(X, Y \in A)
= \|T_{\sqrt{1-\epsilon}} 1_A\|_2^2 \le \|1_A\|^2_{p}
= \P(A)^{\frac{2}{p}}.
\]
Hence, $\P_\epsilon(Y \in A | X \in A) \le \P(A)^{2/p - 1}$.
Taking the logarithm and dividing by $\ln \P(A)$ (which is negative), we have
\[
  M_\epsilon(A) 
= \frac{\ln \P_\epsilon(X, Y \in A)}{\ln \P(A)}
\ge \frac{2}{p} - 1 = \frac{\ln \frac{1}{1-\epsilon}}{\ln s} - \frac{\delta}{\ln s}.
\qedhere
\]
\end{proof}

\section{Hamming ball}

In this section, we consider the example of the Hamming
ball $A_{s,\alpha,n}$ consisting of $x \in [s]^n$ such that
$\#\{i: x_i = 0\} \le \frac{n}{s} - \alpha \sqrt n$.
This is an interesting example because~\cite{BogdanovMossel:11}
showed that if $\alpha$ is sufficiently large (depending on $\epsilon$),
then as $n \to \infty$, $A_{2,\alpha,n}$
achieves the upper bound of
Theorem~\ref{thm:upper-bound}. We will show, however, that this
is no longer true for large $s$.

Note that $1_{X_1 = 0}$ has mean $\frac{1}{s}$ and variance
$\frac{s-1}{s^2}$. Thus, the Berry-Ess\'een theorem implies that
for any fixed $\alpha$ and $s$,
\begin{equation}\label{eq:hamming}
  \P(A_{s,\alpha,n}) \to \P\Big(Z \le -\frac{\alpha s}{\sqrt{s-1}}\Big)
\end{equation}
as $n \to \infty$, where $Z \sim \normal(0, 1)$.
Moreover, if $(Z_1, Z_2) \sim \normal(0,
(\begin{smallmatrix}1 & 1-\epsilon \\ 1-\epsilon & 1\end{smallmatrix}))$
then
\begin{equation}\label{eq:correlated-hamming}
  \P_\epsilon(X , Y \in A_{s,\alpha,n})
  \to \P\Big(Z_1, Z_2 \le -\frac{\alpha s}{\sqrt{s-1}}\Big).
\end{equation}
In particular, by studying normal probabilities we can
use~\eqref{eq:hamming} and~\eqref{eq:correlated-hamming} to compute
$\lim_{n\to\infty} M_\epsilon(A_{s,\alpha,n})$.

\begin{lemma}\label{lem:normal-prob}
  Suppose that $(Z_1, Z_2) \sim \normal(0,
(\begin{smallmatrix}1 & 1-\epsilon \\ 1-\epsilon & 1\end{smallmatrix}))$.
  There is a sufficiently small constant $c$ such that for all $t > 0$ and $0 < \epsilon < 1$,
\[
  \P(Z_1 \ge t \mid Z_2 \ge t) \le \P(Z_1 \ge t)^{c\epsilon}.
\]
\end{lemma}

Lemma~\ref{lem:normal-prob} has the following immediate consequence
for $M_\epsilon(A_{s,\alpha,n})$:
\begin{corollary}\label{cor:hamming}
  There exists a constant $c$ such that for any $s$ and $\alpha$,
  \[
    \lim_{n \to \infty} M_\epsilon(A_{s,\alpha,n}) \ge c \epsilon.
  \]
\end{corollary}

By comparison, the trivial protocol $A = \{x: x_1 = \cdots = x_k = 0\}$
has
\[
  M_\epsilon(A) = \frac{1}{\ln s}
  \ln\bigg(\frac{1}{1 - (1 - s^{-1}) \epsilon}\bigg)
  \le \frac{C' \epsilon}{\ln s}.
\]
In particular, for a fixed success probability and a sufficiently
large alphabet $s$, the trivial protocol recovers $c \ln s$ times
as many symbols as the Hamming ball protocol.

\begin{proof}[Proof of Corollary~\ref{cor:hamming}]
  According to~\eqref{eq:hamming} and~\eqref{eq:correlated-hamming},
  \[
    M_\epsilon(A_{s,\alpha,n}) \to \frac{ 
      \log \P\Big(Z_1 \le -\frac{\alpha s}{\sqrt{s-1}}\Big\mid Z_2 \le -\frac{\alpha s}{\sqrt{s-1}}\Big)
    }{
      \log \P\Big(Z_1 \le -\frac{\alpha s}{\sqrt{s-1}}\Big)
    }.
  \]
  Now apply Lemma~\ref{lem:normal-prob} to the numerator
  (recalling that the denominator is negative):
  \[
    \lim M_\epsilon(A_{s,\alpha,n}) \ge
    \frac{\log \P\Big(Z_1 \le -\frac{\alpha s}{\sqrt{s-1}}\Big)^{c\epsilon}}
    {\log \P\Big(Z_1 \le -\frac{\alpha s}{\sqrt{s-1}}\Big)} = c \epsilon.
    \qedhere
  \]
\end{proof}

\begin{proof}[Proof of Lemma~\ref{lem:normal-prob}]
  The proof makes use of the Ornstein-Uhlenbeck semigroup $P_t$,
  defined by
  \[
    (P_\tau f)(x) = \E f(e^{-\tau}x + \sqrt{1-e^{-2\tau}} Z),
  \]
  where $Z \sim \normal(0, 1)$.
  The Nelson-Gross~\cite{Nelson:73,Gross:75} hypercontractive
  inequality states that
  \begin{equation}\label{eq:gaussian-hypercontractive}
    \big(\E P_\tau |f(Z)|^q\big)^{1/q}
    \le \big(\E |f(Z)|^p\big)^{1/p}
  \end{equation}
  whenever $q \le 1 + e^{2\tau} (p-1)$. If we set $f(x) = 1_{x \ge t}$
  and $\tau = -\log(1-\epsilon)$, then
  \[
    \P(Z_1, Z_2 \ge t) = \E f(Z_1) f(Z_2) = \E f(Z) P_\tau f(Z)
    = \E (P_{\tau/2} f(Z))^2.
  \]
  Thus,~\eqref{eq:gaussian-hypercontractive} with $q=2$
  and $p = 1 + e^{-2\tau} = 1 + (1-\epsilon)^2$ implies that
  \[
    \P(Z_1, Z_2 \ge t) \le (\E f(Z))^{\frac{2}{1 + (1-\epsilon)^2}}
    = \P(Z_2 \ge t)^{\frac{2}{1 + (1-\epsilon)^2}}
    \le \P(Z_2 \ge t)^{1 + c \epsilon}.
  \]
  Hence,
  \[
    \P(Z_1 \ge t | Z_2 \ge t) \le 
    \frac{\P(Z_1 \ge t, Z_2 \ge t)}{\P(Z_2 \ge t)} \le \P(Z_2 \ge t)^{c \epsilon}.
    \qedhere
  \]
\end{proof}

\bibliographystyle{plain}
\bibliography{my,all}

\begin{thebibliography}{10}

\bibitem{BogdanovMossel:11}
A.~Bogdanov and E.~Mossel.
\newblock On extracting common random bits from correlated sources.
\newblock {\em IEEE Transactions on information theory}, 57(10):6351--6355,
  2011.
\newblock Arxiv 1007.2135.

\bibitem{Gross:75}
Leonard Gross.
\newblock Logarithmic {S}obolev inequalities.
\newblock {\em Amer. J. Math.}, 97(4):1061--1083, 1975.

\bibitem{LLGSVD:05}
D.~Lim, J.W. Lee, B.~Gassend, G.E. Suh, M.~Van~Dijk, and S.~Devadas.
\newblock Extracting secret keys from integrated circuits.
\newblock {\em IEEE Transactions on Very Large Scale Integration (VLSI)
  Systems}, 13(10):1200--1205, 2005.

\bibitem{MoOdOl:05}
E.~Mossel, R.~O'Donnell, and K.~Oleszkiewicz.
\newblock Noise stability of functions with low influences: invariance and
  optimality (extended abstract).
\newblock In {\em 46th Annual IEEE Symposium on Foundations of Computer Science
  (FOCS 2005), 23-25 October 2005, Pittsburgh, PA, USA, Proceedings}, pages
  21--30. IEEE Computer Society, 2005.

\bibitem{MORSS:06}
E.~Mossel, R.~O'Donnell, O.~Regev, J.~E. Steif, and B.~Sudakov.
\newblock Non-interactive correlation distillation, inhomogeneous {M}arkov
  chains, and the reverse {B}onami-{B}eckner inequality.
\newblock {\em Israel J. Math.}, 154:299--336, 2006.

\bibitem{MoOlSe:12}
E.~Mossel, K.~Oleszkiewicz, and A.~Sen.
\newblock On reverse hypercontractivity.
\newblock 2011.

\bibitem{Nelson:73}
Edward Nelson.
\newblock The free {M}arkoff field.
\newblock {\em J. Functional Analysis}, 12:211--227, 1973.

\bibitem{Oleszkiewicz:03}
K.~Oleszkiewicz.
\newblock {On a nonsymmetric version of the Khinchine-Kahane inequality}.
\newblock {\em Progress In Probability}, 56:156--168, 2003.

\bibitem{SHO:08}
Y.~Su, J.~Holleman, and B.P. Otis.
\newblock A digital 1.6 p{J}/bit chip identification circuit using process
  variations.
\newblock {\em Solid-State Circuits, IEEE Journal of}, 43(1):69--77, 2008.

\bibitem{Wolff:07}
P.~Wolff.
\newblock Hypercontractivity of simple random variables.
\newblock {\em Studia Mathematica}, pages 219--326, 2007.

\bibitem{Yang:07}
Ke~Yang.
\newblock On the (im)possibility of non-interactive correlation distillation.
\newblock {\em Theoretical Computer Science}, 382(2):157--166, 2007.

\bibitem{YLHMGZ:09}
H.~Yu, P.H.W. Leong, H.~Hinkelmann, L.~Moller, M.~Glesner, and P.~Zipf.
\newblock Towards a unique {FPGA}-based identification circuit using process
  variations.
\newblock In {\em 19th International Conference on Field Programmable Logic and
  Applications}, pages 397--402. IEEE, 2009.

\end{thebibliography}
\end{document}